\documentclass[11pt]{llncs}

\usepackage[latin5]{inputenc}

\usepackage{amssymb}
\usepackage{amsmath}
\usepackage{amsthm}
\usepackage{textcomp}
\usepackage{array}
\usepackage{multirow}
\usepackage{rotating}
\usepackage{color}
\usepackage[left=2.5cm,top=3cm,right=2.5cm,bottom=3cm]{geometry}

\newcommand{\bra}[1]{\langle #1|}
\newcommand{\ket}[1]{|#1\rangle}

\newcommand{\cent}[0]{\mbox{\textcent}}
\newcommand{\dollar}[0]{\$}
\title{Succinctness of two-way probabilistic and quantum finite automata\thanks{A preliminary version of this paper 
was presented at the AutoMathA Plenary Conference 2009, in Li\`{e}ge, Belgium and this 
work was partially supported by the Scientific and Technological
Research Council of Turkey (T\"{U}B\.ITAK) with grant 108142
and the Bo\u{g}azi\c{c}i University Research Fund with grant 08A102.}}
\author{Abuzer Yakary{\i}lmaz\ \and A.C. Cem Say }
\institute{Bo\u{g}azi\c{c}i University, Department of Computer Engineering,\\ Bebek 34342 \.{I}stanbul, Turkey \\
\email{{abuzer,say}@boun.edu.tr}
 \\~~\\
December 22, 2009
}

\begin{document}

\newlength{\twidth}
\maketitle
\pagenumbering{arabic}

%-----------------------------------------------------------------------------%
\begin{abstract} \label{abstract:Abstract}
%-----------------------------------------------------------------------------%
We prove that two-way probabilistic and quantum finite automata
(2PFA's and 2QFA's) can be considerably more concise
than both their one-way versions (1PFA's and 1QFA's), and two-way
nondeterministic finite automata (2NFA's). For this purpose, we demonstrate
several infinite families of regular languages which can be recognized
with some fixed
probability greater than $ \frac{1}{2} $ by just tuning the transition
amplitudes of a 2QFA (and, in one case, a
2PFA) with a constant number of states, whereas the sizes of the
corresponding 1PFA's, 1QFA's and 2NFA's
grow without bound. We also show that 2QFA's with mixed states can
support highly efficient probability amplification. The weakest
known model of computation where quantum computers recognize more
languages with bounded error than their classical counterparts is introduced.
\end{abstract}

%-----------------------------------------------------------------------------%
\section{Introduction} \label{section:Introduction}
%-----------------------------------------------------------------------------%

In recent years, the research effort on quantum versions of finite
automata has mainly focused on one-way models, with the study of
two-way quantum finite automata (2QFA's), which are synonymous with
constant space quantum Turing machines, receiving relatively less
attention. In their seminal paper, Kondacs and Watrous \cite{KW97} proved
that 2QFA's recognize all regular languages with zero error, and the
language $ L_{eq}=\{a^{n}b^{n} \mid n \ge 0\} $
with any desired error bound $ \epsilon > 0 $, in
time $ O(\frac{1}{\epsilon}|w|) $, using $ O(\left( \frac{1}{\epsilon}
\right)^{2} ) $ states,
where $ w $ is the input string.
Since two-way probabilistic finite automata (2PFA's) can decide $
L_{eq} $ only in exponential time
\cite{Fr81,KF90,DS92}, this established the superiority of
2QFA's over 2PFA's. Parallelling work by Aharonov et al. \cite{AKN98}
on quantum circuits with mixed states, Ambainis and
Watrous \cite{AW02} introduced an alternative model, the two-way finite
automaton with quantum and classical states (2QCFA), which includes a
constant-size quantum part which may be in a mixed state, but requires
the tape head position to be classical. Yakary{\i}lmaz and Say
\cite{YS09B} noted that
conventional methods of probability amplification give significantly
inefficient results when applied to 2QFA's, and presented methods
which can be used to decide $ L_{eq} $ with error bound $ \epsilon $
in
as low as $ O(|w|) $ steps (i.e. with runtime independent of $ \epsilon $),
and with as low as $ O(\log^{2}(\frac{1}{\epsilon})\log\log(\frac{1}{\epsilon})) $ states.

Issues of succinctness, as exemplified above, constitute a rich
subtopic of automata theory \cite{DS90,KF90,AF98,MPP01,ANTV02,FOM09,YS09B}.
In this paper, we examine how the combination of two-wayness
and (quantum or classical) probabilistic transitions affects
succinctness. As our main result, we demonstrate several infinite families
of regular languages which can be recognized with some fixed
probability greater than $ \frac{1}{2} $ by just tuning the transition
amplitudes of a 2QFA (and, in one case, a 2PFA) with a constant number
of states, whereas the sizes of the corresponding one-way machines,
and two-way nondeterministic finite automata (2NFA's) grow without bound.

The Kondacs-Watrous model of quantum finite automaton (to be
called, from now on, KWQFA), which allows measurements of a restricted
type, rather than the full set sanctioned by quantum theory, has been
proven to be weaker in terms of language recognition power \cite{KW97},
probability amplification capability \cite{AF98}, and, in
some cases at least, succinctness \cite{ANTV02}, than the
corresponding classical model, in the one-way case. More general
models, such as the 2QCFA, employing mixed
states, are able to simulate the corresponding classical probabilistic
automata efficiently in both the one-way and two-way settings, and
to recognize some languages that 2PFA's cannot \cite{AW02}. We show
that 2QFA's with mixed
states can
support highly efficient probability amplification, surpassing the
best known methods for 2KWQFA's recognizing these languages.

We introduce a new model of quantum automaton,
named the two-way quantum finite automaton with reset.
This is an enhancement to the 2KWQFA, endowing it with the capability of
resetting the position of the tape
head to the left end of the tape in a single move during the computation.
We use this model both in the proof of our main result, and in the
demonstration of our probability amplification techniques. We mostly
focus on a restricted form of these machines, called the one--way
quantum finite automaton with restart (1QFA$^{\circlearrowleft}$), which can switch only to the
initial state during left resets, and cannot perform single-step left
or stationary moves. We give evidence that this is the weakest
known model of computation where quantum computers recognize more
languages with bounded error than their classical counterparts.

The rest of this paper is structured as follows: Section 2 contains
the definitions and some basic facts about our new model that will be
used throughout the paper. In Section 3, we prove some key lemmata about 
the relationship between one--way quantum finite automata with and without restart, 
and examine the class of languages recognized with bounded error by 1QFA$^{\circlearrowleft}$'s.
Our main succinctness result is presented
in Section 4. In Section 5, we present several algorithms that improve
previous results about the efficiency of probability amplification in
2KWQFA's and 2QCFA's. In Section 6, we investigate the
computational power of probabilistic finite
automata with restart. Section 7 is a
conclusion.

%-----------------------------------------------------------------------------%
\section{Preliminaries} \label{section:Preliminaries}
%-----------------------------------------------------------------------------%

Watrous \cite{Wa97} notes that a 2KWQFA
algorithm he presents for recognizing a nonregular language is
remarkably costly in terms of probability amplification, and states
that this problem stems from the fact that 2KWQFA's cannot ``reset"
themselves during execution to repeatedly carry out the same
computation. The 2QCFA model provides one way of solving this problem,
by having a classical part, in addition to the quantum register. We
present an alternative 2QFA model, employing only quantum states,
whose only difference from the 2KWQFA is the existence of an
additional ``reset move" in its repertory. Section 2.1 contains the
definitions of this and the other models that will be examined in the
paper. Section 2.2 describes some facts which will make the analyses
of the algorithms in later sections easier.

%***************************************************%
\subsection{Definitions} \label{subsection:definitions}
%***************************************************%
Let $ \Sigma $ be an input alphabet, not containing the end--marker
symbols $ \cent $ and $ \dollar $, and let $ \Gamma = \Sigma \cup
\{\cent, \dollar\} $ be the tape alphabet.

A \textit{2-way quantum finite automaton with reset} (2QFA$
^{\curvearrowleft} $) is a 7-tuple
\begin{equation}
    \label{equation:2qfa-with-reset-tuple}
   \mathcal{M}=(Q,\Sigma,\delta,q_{0},Q_{acc},Q_{rej},
   Q_{reset}=\cup_{ q \in Q_{non}} Q^{\curvearrowleft}_{q} ),
\end{equation}
where
\begin{enumerate}
      \item $ Q = \{q_{0},\ldots,q_{n} \} $ is the finite set of states;
      \item $ \delta $ is the transition function, described below;
      \item $ q_{0} \in Q $ is the initial state;
      \item $ Q_{acc} $ is the set of accepting states;
      \item $ Q_{rej} $ is the set of rejecting states;
		\item $ Q_{non} = Q \setminus (Q_{acc} \cup Q_{rej} \cup Q_{reset}) $ is the set of nonhalting
     and nonresetting states;
      \item $ Q_{reset} $ is the union of disjoint reset sets, i.e., each $ Q_{q \in Q_{non}}^{\curvearrowleft} $
     contains reset states that cause the computation to restart with state $ q $, as described below.
\end{enumerate}

We assume that the states in $ Q_{non} $ have smaller indices than
other members of $Q$; $ q_{i} \in Q_{non} $ for $ 0 \le i < | Q_{non} | $.

The configurations of a 2QFA$ ^{\curvearrowleft} $ are pairs of the form
$(state,head~position)$. Initially, the head is on the left end-marker $\cent$,
and so the machine starts computation in the superposition $ \ket { q_{0},0 }$.

The transition function of a 2QFA$ ^{\curvearrowleft} $ working on an
input string $ w \in \Sigma^{*}
$, (that is, a tape containing $ \mathsf{w}= \cent w \dollar $,)
is required to induce a unitary operator $U_{\delta}^{w}$ on the
Hilbert space $\ell_2(Q \times \mathbb{Z}_{|\mathsf{w}|} )$, since
quantum machines can exist in superpositions of more than one
configuration.

In all 2QFA$ ^{\curvearrowleft} $'s described in this paper, every
transition entering the same state  involves the tape
head moving in the same direction (left, right, or stationary).
With this simplification, considering the Hilbert space $\ell_2(Q)$,
a syntactically correct 2QFA$ ^{\curvearrowleft} $ (that is, one where
$U_{\delta}^{w}$ is unitary for every $w$,) can be specified
easily by just providing
a unitary operator $\mathsf{U}_{\sigma}$ : $\ell_2(Q) \rightarrow
\ell_2(Q)$ for each symbol $\sigma\in\Gamma$.
More formally,
\begin{equation}
   \delta(q,\sigma,q^{'},d_{q^{'}})=\bra{q^{'}}{\mathsf{U}_{\sigma}}\ket{q}
\end{equation}
is the amplitude with which the machine currently in state $q$ and
scanning symbol $\sigma$
will jump to state $q^{'}$ and move the head in direction $d_{q^{'}}$.
Here, $d_{q^{'}}\in\{-1,0,1\}$ is the direction of the tape head
determined by $q^{'}$. For the remaining directions, all transitions
with
target $q^{'}$ have amplitude zero.

Apart from the left reset capability, 2QFA$ ^{\curvearrowleft} $'s are
identical to 2KWQFA's. In the following, we focus on this new
capability, and refer the reader to \cite{KW97} for detailed coverage
of the technical properties of 2KWQFA's.

In each step of its execution, a 2QFA$ ^{\curvearrowleft} $ undergoes two linear
operations: The first one is a unitary transformation of the current
superposition according to $\delta$, and the second one is a
measurement. The observable
describing this measurement process is designed so that the outcome of
any observation is one of ``accept",
``reject", ``continue without resetting", or ``reset with state $ q
$", for any $ q
\in Q_{non} $.
Formally, we use the observable $ \mathcal{O} $, corresponding to the
decomposition
\begin{equation}
	\label{equation:Epartition}
     E=E_{acc} \oplus E_{rej} \oplus E_{non} \oplus E_{reset-0} \oplus E_{reset-1} 
     \oplus \cdots \oplus E_{reset-(k-1)},
\end{equation}
where $ k=|Q_{non}|$, and for a given input $ w$,
\begin{enumerate}
    \item the set of all configurations of the 2QFA$
^{\curvearrowleft} $ is $ Q \times \mathbb{Z}_{|\mathsf{w}|} $;
    \item $ E = \mbox{span} \{\ket{c} \mid c \in Q \times
\mathbb{Z}_{|\mathsf{w}|} \} $;
    \item $ E_{acc} = \mbox{span} \{\ket{c} \mid c \in Q_{acc}
\times \mathbb{Z}_{|\mathsf{w}|} \} $;
    \item $ E_{rej} = \mbox{span} \{\ket{c} \mid c \in Q_{rej}
\times \mathbb{Z}_{|\mathsf{w}|} \} $;
    \item $ E_{non} = \mbox{span} \{\ket{c} \mid c \in Q_{non}
\times \mathbb{Z}_{|\mathsf{w}|} \} $;
    \item $ E_{reset-i} = \mbox{span} \{\ket{c} \mid c \in Q_{q_{i}
\in Q_{non}}^{\curvearrowleft} \times
                     \mathbb{Z}_{|\mathsf{w}|} \} $ $ (0 \le i \le k-1) $.
\end{enumerate}
The probability of each outcome
is determined by the amplitudes of the relevant configurations in the
present superposition. The contribution of each configuration
to this probability is the modulus squared of its amplitude. 
For instance, the outcome ``reset with state $q_i$" will be measured with probability
$\sum_{c \in Q_{q_{i}}^{\curvearrowleft} \times \mathbb{Z}_{|\mathsf{w}|}} |\alpha_c|^2$,
where $\alpha_c$ is the amplitude of configuration $c$.
If ``accept" or ``reject" is measured, the computation halts.
If ``continue without resetting" is measured, the machine continues
running from a
superposition of the nonhalting and nonresetting configurations, obtained by
normalizing the projection of the superposition before the measurement
onto $ \mbox{span} \{\ket{c} |
c \in Q_{non} \times \mathbb{Z}_{|\mathsf{w}|} \} $. If
``reset with state $ q_i $" is measured, the tape head is reset to point
to the left end-marker, and the machine continues from the
superposition $ \ket { q_i,0 }$ in the next step. Note that the decoherence
associated with this measurement means that the system allows mixed states.

A 2QFA$ ^{\curvearrowleft} $ $ \mathcal{M} $ is said to recognize a
language $ L $ with error bounded by
$\epsilon$ if $ \mathcal{M} $'s computation results in ``accept" being
measured for all members of $ L $
with probability at least $ 1-\epsilon $, and ``reject" being measured
for all other inputs with probability at
least $1-\epsilon$.

A \textit{2-way quantum finite automaton with restart} (2QFA$
^{\circlearrowleft} $) is a
restricted 2QFA$ ^{\curvearrowleft} $ in which the ``reset moves" can
target only the original start state of
the machine, that is, in terms of Equation \ref{equation:2qfa-with-reset-tuple},
all the $ Q^{\curvearrowleft}_{q} $ of a 2QFA$ ^{\circlearrowleft} $ are
empty, with the exception of $ Q^{\curvearrowleft}_{q_{0}} $,
represented as $ Q_{restart} $.

The \textit{two-way probabilistic finite automaton} (2PFA) is the
classical probabilistic counterpart of 2KWQFA's; see \cite{Ka91}
for the details. A \textit{one-way probabilistic finite automaton}
(1PFA) \cite{Ra63} is a 2PFA in which the head moves
only to the right in every step. A \textit{rational 1PFA}  \cite{Tu69b} 
is a 1PFA where all entries in the transition matrices are rational numbers.

Other variants of two-way automata with reset that will be examined
in this paper are
\begin{enumerate}
     \item A \textit{one-way (Kondacs-Watrous) quantum finite
automaton with reset}
     (1QFA$ ^{\curvearrowleft} $) is a restricted 2QFA$
^{\curvearrowleft} $ which uses neither ``move
     one square to the left" nor ``stay put" transitions, and whose
tape head is therefore classical,
  \item A \textit{one-way (Kondacs-Watrous) quantum finite
automaton with restart} (1QFA$ ^{\circlearrowleft} $)
  is a 1QFA$ ^{\curvearrowleft} $ where the reset moves can target
only the original start state, and,
     \item A \textit{one-way probabilistic finite automaton with
restart} (1PFA$ ^{\circlearrowleft} $) is a 1PFA
     which has been enhanced with the capability of resetting the
tape head to the left
      end-marker and swapping to the original start state.
\end{enumerate}

A \textit{one-way (Kondacs-Watrous) quantum finite
automaton} (1KWQFA) \cite{KW97} is a 2KWQFA  which  moves its tape head
only to the right in every step.

A well--known two-way mixed--state model is the 2QCFA \cite{AW02}.
Formally, a \textit{2-way finite automaton with quantum and classical
states} (2QCFA) is a 9-tuple
\begin{equation}
	\label{equation:2qcfa}
	\mathcal{M}=(Q,S,\Sigma,\mathsf{\Theta},\delta,q_{0},s_{0},S_{acc},S_{rej}),
\end{equation}
where
\begin{enumerate}
    \item $ Q = \{q_{0},\ldots,q_{n_{1}} \} $ is the finite set of the quantum states;
    \item $ S = \{s_{0},\ldots,s_{n_{2}} \} $ is the finite set of the classical states;
    \item $ \mathsf{\Theta} $ and $ \delta $ govern the machine's behavior, as described below;
    \item $ q_{0} \in Q $ is the initial quantum state;
    \item $ s_{0} \in S $ is the initial quantum state;
    \item $ S_{acc} \subset S $ is the set of classical accepting states;
    \item $ S_{rej} \subset S $ is the set of classical rejecting states.
\end{enumerate}

The functions $ \mathsf{\Theta} $ and $ \delta $ specify the evolution
of the quantum and classical parts of $ \mathcal{M} $, respectively.
Both functions take the currently scanned symbol $ \sigma \in \Gamma $
and current classical state $ s \in S $ as arguments.
$ \mathsf{\Theta}(s,\sigma) $ is either a unitary transformation, or an
orthogonal measurement.
In the first case, the new classical state and tape head direction 
(left, right, or stationary)
are determined by $ \delta $, depending on $s$ and $\sigma $. In the
second case, when an orthogonal measurement is applied on the quantum
part, $ \delta $ determines the new classical state and the tape head
direction using the result
of that measurement, as well as $s$ and $\sigma $.
The quantum and classical parts are initialized with $ \ket{q_{0}} $
and $ s_{0} $, respectively, and the tape head
starts on the first cell of the tape, on which $ \cent w \dollar $ is written
for a given input string $ w \in \Sigma^{*} $.
During the computation, if an accepting or rejecting
state is entered, the machine halts with the relevant response to the
input string.

Note that like the 1QFA$ ^{\curvearrowleft} $, and unlike the 2QFA
and the 2QFA$ ^{\curvearrowleft} $,
the tape head position of a 2QCFA is classical, (that is, there are no
superpositions with the head in more than one position
simultaneously,) meaning that the machine can be
implemented using a quantum part of constant size.

%***************************************************%
\subsection{Basic facts} \label{subsection:basicfacts}
%***************************************************%

We start by stating some basic facts concerning automata with restart, which
will be used in later sections.

A segment of computation which begins with a (re)start, and ends with
a halting or restarting configuration will be called a \textit{round}.
Clearly, every automaton with restart which makes nontrivial use of
its restarting capability will run for infinitely many rounds on some
input strings. Throughout this paper, we make the assumption that our
two-way automata do not contain infinite loops within a round, that
is, the computation restarts or halts with probability 1
in a finite number steps for each round.

Everywhere in this section, $ \mathcal{R} $ will stand for a finite
state automaton with restart, and $ w \in
\Sigma^{*} $ will represent an input string using the alphabet $ \Sigma $.

\begin{definition}
	\label{definition:acc-rej-halt}~
 \begin{list}{$ \bullet $}{}
        \item $ p_{acc}(\mathcal{R},w) $, $
p_{rej}(\mathcal{R},w) $, and $
p_{restart}(\mathcal{R},w) $
        denote the probabilities that
        $ \mathcal{R} $ will accept, reject, or restart, respectively, on
input $ w $, in the first round.
        \item $ \mathsf{P_{acc}}(\mathcal{R},w) $ and $
\mathsf{P_{rej}}(\mathcal{R},w) $
        denote the overall acceptance and rejection
probabilities of $ w $
by $ \mathcal{R} $, respectively.
 \end{list}
 Moreover, $
p_{halt}(\mathcal{R},w)=p_{acc}(\mathcal{R},w)+p_{rej}(\mathcal{R},w) $.
\end{definition}
\begin{lemma}
 \label{lemma:acc-rej}
 \begin{equation}
        \label{equation:Paccw-and-Prejw}
        \mathsf{P_{acc}}(\mathcal{R},w)=\dfrac{1}{1+\frac{p_{rej}(\mathcal{R},w)}{p_{acc}(\mathcal{R},w)}};~~
        \mathsf{P_{rej}(\mathcal{R},w)}=\dfrac{1}{1+\frac{p_{acc}(\mathcal{R},w)}{p_{rej}(\mathcal{R},w)}}.
 \end{equation}
\end{lemma}
\begin{proof}
 \begin{eqnarray*}
       \mathsf{P_{acc}}(\mathcal{R},w) & = &
        \sum_{i=0}^{\infty}\left(1-p_{acc}(\mathcal{R},w)-p_{rej}(\mathcal{R},w)\right)^{i}p_{acc}(\mathcal{R},w)\\
               & = & p_{acc}(\mathcal{R},w)\left(
 \dfrac{1}{1-(1-p_{acc}(\mathcal{R},w)-p_{rej}(\mathcal{R},w))}
\right) \\
               & = &
\dfrac{p_{acc}(\mathcal{R},w)}{p_{acc}(\mathcal{R},w)+p_{rej}(\mathcal{R},w)}
\\
               & = &
\dfrac{1}{1+\frac{p_{rej}(\mathcal{R},w)}{p_{acc}(\mathcal{R},w)}}.
 \end{eqnarray*}
 $ \mathsf{P_{rej}}(\mathcal{R},w) $ is calculated in the same way.
\end{proof}

\begin{lemma}
 \label{lemma:prej-over-pacc}
 The language $ L \subseteq \Sigma^{*} $ is recognized by $ \mathcal{R} $ with error bound $ \epsilon >0 $  
 if and only if $ \frac{p_{rej}(\mathcal{R},w)}{p_{acc}(\mathcal{R},w)} \le
\frac{\epsilon}{1-\epsilon} $  when $ w \in L $, and $ \frac{p_{acc}(\mathcal{R},w)}{p_{rej}(\mathcal{R},w)} \le
\frac{\epsilon}{1-\epsilon} $ when $ w \notin L $.
Furthermore, if $ \frac{p_{rej}(\mathcal{R},w)}{p_{acc}(\mathcal{R},w)} $ 
$ ( \frac{p_{acc}(\mathcal{R},w)}{p_{rej}(\mathcal{R},w)} ) $ is at most $ \epsilon $, then 
$ \mathsf{P_{acc}}(\mathcal{R},w) $ $ ( \mathsf{P_{rej}}(\mathcal{R},w) ) $ is at least $ 1-\epsilon $.
\end{lemma}
\begin{proof}
 This follows from Lemma \ref{lemma:acc-rej}, since, for all $ p \ge 0 $, $ \epsilon \in [0,\frac{1}{2}) $,
 \begin{equation}
 \label{equation:1overp-epsilon}
       \frac{1}{1+p} \ge 1-\epsilon \Leftrightarrow p \le
\frac{\epsilon}{1-\epsilon}, \mbox{ and }
 \end{equation}
 \begin{equation}
 \label{equation:1overp-epsilon-simplified}
       p \le \epsilon \Rightarrow \frac{1}{1+p} \ge 1-\epsilon.
 \end{equation}
\end{proof}

\begin{lemma}
 \label{lemma:expected-runtime}
 Let $ p=p_{halt}(\mathcal{R},w) $, and let $ s(w) $ be the maximum number
of steps in any branch of a
 round of $ \mathcal{R} $ on $ w $.
 The worst-case expected runtime of $ \mathcal{R} $ on $ w $ is
 \begin{equation}
       \label{equation:expected-runtime}
       \frac{1}{p} (s(w)).
 \end{equation}
\end{lemma}
\begin{proof}
 The worst-case expected running time of $ \mathcal{R} $ on $ w $ is
    \begin{equation}
            \label{equation:expected-runtime-proof}
            \sum_{i=0}^{\infty} (i+1)(1-p)^{i} (p)(s(w)) =
       (p)(s(w))\frac{1}{p^{2}}=\frac{1}{p}(s(w)).
    \end{equation}
\end{proof}

\begin{lemma}
	\label{lemma:restart-time}
	Any one-way automaton with restart with expected runtime 
	$ t $ can be simulated by a corresponding two-way automaton without restart 
	in expected time no more than $ 2t $.
\end{lemma}
\begin{proof}
The program of the two-way machine ($ \mathcal{R}_{2} $) is identical
to that of the one-way machine with restart ($ \mathcal{R}_{1} $),
except for the fact that each restart move of $ \mathcal{R}_{1} $ is
imitated by $ \mathcal{R}_{2} $ by moving the head one square per step
all the way to the left end-marker. This causes the runtimes of the
$i$ nonhalting rounds in the summation in Equation
(\ref{equation:expected-runtime-proof})
in Lemma \ref{lemma:expected-runtime} to increase by a factor of 2.
\end{proof}

We will now give a quick review of the technique of probability
amplification. Suppose that we are given a machine (with or without reset) $
\mathcal{A} $,
which recognizes a language $L$ with error bounded by $\epsilon$, and
we wish to construct another machine which recognizes $L$ with a much
smaller, but still positive, probability of error, say, $ \epsilon^{'}
$. It is well known\footnote{See, for instance, pages 369-370 of
\cite{Si06}.} that one can achieve this by running $ \mathcal{A} $  $
 O(\log(\frac{1}{\epsilon^{'}})) $ times on the same input, and then
giving the majority answer as our verdict about the membership of the
input string in $L$.

Suppose that the original machine $ \mathcal{A} $ needs
to be run $ 2k+1 $ times for the overall procedure to work with the
desired correctness probability. Two counters can be used to count the
acceptance and rejection responses, and
the overall computation
accepts (rejects) when the number of recorded acceptances (rejections)
reaches $ k+1 $.
To implement these counters in the finite automaton setting, we need
to ``connect" $ (k+1)^{2} $ copies of $ \mathcal{A} $, $ \{
\mathcal{A}_{i,j} \mid 0 \le i,j \le k \} $,
where the subscripts indicate the values of the two counters,
i.e., the states of $ \mathcal{A}_{i,j} $ encode the information that
$ \mathcal{A} $
has accepted $ i $ times and rejected $ j $ times in its previous
runs. The new machine $ \mathcal{M} $ is constructed
from the $ \mathcal{A}_{i,j} $ as follows:
\begin{itemize}
 \item The start state of $ \mathcal{M} $ is the start state of $
\mathcal{A}_{0,0} $;
 \item Upon reaching any accept state of $ \mathcal{A}_{i,j} $ ($
0 \le i,j < k $), $ \mathcal{M} $ moves the head back to the left
end-marker and then switches to the start state of $
\mathcal{A}_{i+1,j} $;
    \item Upon reaching any reject states of $ \mathcal{A}_{i,j} $ ($ 0
\le i,j < k $), $ \mathcal{M} $ moves the head back to the left
end-marker and then switches to the start state of $
\mathcal{A}_{i,j+1} $;
 \item The accept states of $ \mathcal{M} $ are the accept states
of $ \mathcal{A}_{k,j} $ ($ 0 \le j < k $);
 \item The reject states of $ \mathcal{M} $ are the reject states
of $ \mathcal{A}_{i,k} $ ($ 0 \le i < k $).
\end{itemize}
\begin{lemma}
 \label{lemma:restart-to-reset}
 If language $ L \subseteq \Sigma^{*} $ is recognized by
$\mathcal{R} $ with a fixed error
 bound $ \epsilon > 0 $, then for any positive error bound $
\epsilon^{'} < \epsilon$, there exists a finite automaton with reset,
 $ \mathcal{R}^{'} $, recognizing $ L $ .
 Moreover, if $ \mathcal{R} $ has $ n $ states and its (expected)
runtime is $ O(s(|w|)) $,
 then $ \mathcal{R}^{'} $ has $
O(\log^{2}(\frac{1}{\epsilon^{'}})n) $ states, and
 its (expected) runtime is $ O(\log(\frac{1}{\epsilon^{'}})s(|w|)) $,
 where $ w $ is the input string.
\end{lemma}
\begin{proof}
 Follows easily from the above description.
\end{proof}

Finally, we note the following relationship between the
computational powers of the 2QCFA and the 1QFA$^{\curvearrowleft}$.

\begin{lemma}
 \label{lemma:1qfareset-simulatedby-2qcfa}

For any 1QFA$ ^{\curvearrowleft} $ $
\mathcal{M}_{1} $  with $ n $ states and expected runtime
 $ t(|w|) $, there exists a 2QCFA $ \mathcal{M}_{2} $ with $ n $
quantum states, $ O(n) $ classical states,
 and expected runtime $ O(t(|w|)) $, such that $ \mathcal{M}_{2} $
accepts every input string $w$ with the same probability that $
\mathcal{M}_{1} $ accepts $w$.
\end{lemma}
\begin{proof}
We utilize the 2QCFA's ability of making arbitrary orthogonal measurements.
Given a 1QFA$ ^{\curvearrowleft} $ $ \mathcal{M}_{1} $,
we construct a 2QCFA $ \mathcal{M}_{2} $ with the same set of quantum states.
On each tape square, $ \mathcal{M}_{2} $ first performs the unitary
transformation
associated with the current symbol by the  program of $ \mathcal{M}_{1} $.
It then makes a measurement (over the space spanned by the set of
quantum states) using an observable $ \mathcal{O}^{'} $,
which is formed by replacing each subspace of the form $E_{reset-i}$ in
the observable $ \mathcal{O} $ (Equation \ref{equation:Epartition}) of
$ \mathcal{M}_{1} $\footnote{Since the head is
classical, the observable is redefined to be a decomposition of the
space spanned by just the set of states.} by its subspaces
\[ \{ E_{reset-i-q_{i_1}} \oplus E_{reset-i-q_{i_2}} \oplus \cdots
\oplus E_{reset-i-q_{i_m}} \}, \]
where $ \{ q_{i_1}, q_{i_2}, \cdots, q_{i_m} \} = {
Q^{\curvearrowleft}_{q_{i}} } $, and
$ E_{reset-i-q_{i_j}} = \mbox{span} \{\ket{q_{i_j}}\} $.
The outcome associated with $ E_{reset-i-q_{i_j}} $ is simply the name
of $ q_{i_j}$.

$ \mathcal{M}_{2} $  takes the action specified below according to the result of
this observation:
 \begin{enumerate}
       \item ``continue without resetting": move the head one
square to the right,
       \item ``accept": accept,
       \item ``reject": reject,
       \item ``$ q_{i_j}$": enter a classical
state that moves the head left until the
       left end-marker $ \cent $ is seen, and perform a unitary
transformation that transforms the quantum register from state $
q_{i_j}$ to
       $ q_{i} $.
 \end{enumerate}
\end{proof}

%-----------------------------------------------------------------------------%
\section{Computational power of 1QFA$ ^{\circlearrowleft} $'s} \label{section:1qfa-restart}
%-----------------------------------------------------------------------------%

In this section, we focus on the 1QFA$ ^{\circlearrowleft} $, which turns out to be the simplest and 
most restricted known model of 
quantum computation that is strictly superior in terms of bounded-error language recognition to its classical 
counterpart.

Our first result shows that 1QFA$ ^{\circlearrowleft} $'s can
simulate any 1PFA$ ^{\circlearrowleft} $ with small state cost, albeit with great slowdown.
Note that no such relation is known between the 2KWQFA and its classical counterpart, the 2PFA.

\begin{theorem}
	\label{theorem:pfa-restart-simulated-by-qfa-restart}
	Any language $ L \subseteq \Sigma^{*} $ recognized by an $ n $-state 1PFA$ ^{\circlearrowleft} $
	with error bound $ \epsilon $ can be recognized by a $ 2n+4 $-state 1QFA$ ^{\circlearrowleft} $
	with the same error bound. 
	Moreover, if the expected runtime of the 1PFA$ ^{\circlearrowleft} $ is $ O(s(|w|)) $,
	then the expected runtime of the 1QFA$  ^{\circlearrowleft} $ is $ O(l^{2|w|}s^{2}(|w|)) $
	for a constant $ l>1 $ depending on $n$, where $ w $ is the input string.
\end{theorem}
\begin{proof}
	Let $ \mathcal{P} $ be an $ n $-state 1PFA$ ^{\circlearrowleft} $ recognizing $ L $ 
	with error bound $ \epsilon $. We will construct a $ 2n+4 $-state 1QFA$ ^{\circlearrowleft} $ 
	$ \mathcal{M} $ recognizing the same language with error bound $ \epsilon^{'} \le \epsilon $.

	By adding two more states, $ s_{acc} $ and $ s_{rej} $, to $ \mathcal{P} $, 
	we obtain a new 1PFA$ ^{\circlearrowleft} $, $ \mathcal{P}^{'} $, 
	where the halting of the computation in each round is postponed to the last symbol, $ \dollar $, 
	on which the overall accepting and rejecting probabilities are summed up into 
	$ s_{acc} $ and $ s_{rej} $, respectively.
	Therefore, for any given input string $ w \in \Sigma^{*} $, 
	the value of $ s_{acc} $ and $ s_{rej} $ are $ p_{acc}(\mathcal{P},w) $ and $ p_{rej}(\mathcal{P},w) $,
	respectively, at the end of the first round.

	By using the method described in \cite{YS09D}, each stochastic matrix can be converted to a unitary one
	with twice the size as shown in the template
	\[
		\mathsf{U}=\left(
			\begin{array}{c}
				\frac{1}{l} \left( \mathsf{A} \mid B \right)~
				\\ \hline
				D
			\end{array}
			\right),
	\]
	where $ \mathsf{A} $ is the original stochastic matrix; 
	the columns of $ B $, corresponding to newly added states, are filled 
	in to ensure that each row of $ (A \mid B) $ is pairwise orthogonal
	to the others, and has the same length $ l $, which depends only on the dimension of $ \mathsf{A} $, and
	the entries of $ D $ are then selected to make $ \mathsf{U} $ a unitary matrix.

	Each transition matrix of $ \mathcal{P}^{'} $ can be converted to a $ (2n+4) \times (2n+4) $-dimensional
	unitary matrix according to this template. These are the transition matrices of $ \mathcal{M} $. 
	The state set of $ \mathcal{M} $ can be specified as follows:
	\begin{enumerate}
		\item The states corresponding to $ s_{acc} $ and $ s_{rej} $ are
		the accepting and rejecting states, $ q_{acc} $ and $ q_{rej} $, respectively,
		\item the states corresponding to the non-halting and non-restarting states of 
		$ \mathcal{P}^{'} $ are non-halting and non-restarting states, and,
		\item all remaining states are restarting states.
	\end{enumerate}
	The initial state of $ \mathcal{M} $ is the state corresponding to the initial state of $ \mathcal{P} $.

	When $ \mathcal{M} $ runs on input string $|w|$, the amplitudes of $ q_{acc} $ and $ q_{rej} $, 
	the only halting states of $ \mathcal{M} $, at the end of the first round are
	$ \left( \frac{1}{l} \right)^{|w|+2}p_{acc}(\mathcal{P},w)  $ and
	$ \left( \frac{1}{l} \right)^{|w|+2}p_{rej}(\mathcal{P},w)  $, respectively.
	Therefore, when $ w \in L $,
	\[
		\frac{p_{rej}(\mathcal{M},w)}{p_{acc}(\mathcal{M},w)} =
		\frac{p^{2}_{rej}(\mathcal{P},w)}{p^{2}_{acc}(\mathcal{P},w)}
		\leq
		\frac{\epsilon^{2}}{(1-\epsilon)^{2}},
	\]
	and similarly, when $ w \notin L $,
	\[
		\frac{p_{acc}(\mathcal{M},w)}{p_{rej}(\mathcal{M},w)} =
		\frac{p^{2}_{acc}(\mathcal{P},w)}{p^{2}_{rej}(\mathcal{P},w)}
		\leq
		\frac{\epsilon^{2}}{(1-\epsilon)^{2}}.
	\]
	By solving the equation
	\[ \frac{\epsilon^{'}}{1-\epsilon^{'}} = \frac{\epsilon^{2}}{(1-\epsilon)^{2}}, \]
	we obtain
	\[ \epsilon^{'}=\frac{\epsilon^{2}}{1 - 2\epsilon + 2\epsilon^{2}} \leq \epsilon. \]

	The expected runtime of $ \mathcal{P} $ is 
	\[ \frac{1}{p_{acc}(\mathcal{P},w)+p_{rej}(\mathcal{P},w)} \in O(s(|w|)), \]
	and so the expected runtime of $ \mathcal{M} $ is
	\[ 
		\left( l \right)^{2|w|+4} \frac{1}{p^{2}_{acc}(\mathcal{P},w)+p^{2}_{rej}(\mathcal{P},w)} < 3 
		\left( l \right)^{2|w|+4} \left( \frac{1}{p_{acc}(\mathcal{P},w)+p_{rej}(\mathcal{P},w)} \right)^{2}
		\in O(l^{2|w|}s^{2}(|w|)).
	\]
\end{proof}

\begin{corollary}
     1QFA$  ^{\circlearrowleft} $'s can recognize all regular languages with zero error.
\end{corollary}

To establish the strict superiority of 1QFA$ ^{\circlearrowleft}$'s over 1PFA$ ^{\circlearrowleft} $'s, 
we will make use of the following concepts.

An automaton $ \mathcal{M} $ is said to recognize a language $L$ with
\textit{positive one-sided unbounded error} if every input string $ w \in L $
is accepted by $ \mathcal{M} $ with nonzero
probability, and every $ w \notin L $ is rejected by $ \mathcal{M} $
with probability 1.
An automaton $ \mathcal{M} $ is said
to recognize a language $L$ with \textit{negative one-sided unbounded error}
if every input string $ w \in L $ is accepted by $ \mathcal{M} $
with probability 1,
and every $ w \notin L $ is rejected by $ \mathcal{M} $ with nonzero
probability.

For an automaton $ \mathcal{M} $ recognizing a language $ L $, we define
the \textit{gap function}, $ g_{\mathcal{M}} : N \rightarrow [0, 1] $,
such that $ g_{\mathcal{M}}(n)$ is the
difference between the minimum
acceptance probability of a member of $ L $ with length at most $ n $ and
the maximum acceptance probability of a non-member of $ L $ with length
at most $ n $\footnote{The definition of $ g_{\mathcal{M}} $ is
due to Bertoni and Carpentieri \cite{BC01A}, who call it the ``error function."}.

\begin{lemma}
	\label{lemma:kwqfa-to-1qfarestart-exponential}
	If a language $ L $ is recognized by a 1KWQFA $ \mathcal{M} $ with positive (negative)
	one-sided unbounded error such that $ g_{\mathcal{M}}(n) \geq c^{-n} $ for some $ c>1 $,
	then for all $ \epsilon \in (0,\frac{1}{2}) $, $ L $ is recognized by some 
	1QFA$ ^{\circlearrowleft} $ having three more states than $ \mathcal{M} $ with positive (negative) 
	one-sided error $ \epsilon $
	in expected time $ O(\frac{1}{\epsilon}c^{|w|}|w|) $.	
\end{lemma}
\begin{proof}
	We consider the case of positive one-sided error. The adaptation to the other case is trivial.
	$ \mathcal{M} $ is converted into a 1QFA$ ^{\circlearrowleft} $ $\mathcal{M}^{'}_{\epsilon} $
	as follows.
	$ \mathcal{M}^{'}_{\epsilon} $ starts by branching to two equiprobable paths,
	$ \mathsf{path_{1}} $ and $ \mathsf{path_{2}} $, at the beginning of the computation.
	$ \mathsf{path_{1}} $ imitates the computation of $ \mathcal{M} $, 
	except that all reject states that appear in its subpaths are replaced
	by restart states. Regardless of the form of the input, $ \mathsf{path_{2}} $ moves right with amplitude $
	\frac{1}{\sqrt{c}} $, (and so restarts the computation with the remaining probability,) on every input symbol. 
	When it arrives at the right end-marker, $ \mathsf{path_{2}} $ rejects with amplitude $ \sqrt{\epsilon} $, and
	restarts the computation with amplitude $ \sqrt{1-\epsilon} $. 	

	When $ w \notin L $,
	\[
		p_{acc}(\mathcal{M}^{'}_{\epsilon},w) = 0,
		\mbox{ and }
		p_{rej}(\mathcal{M}^{'}_{\epsilon},w) = \frac{\epsilon}{2c^{|w|}},
    \]
	and so the input is rejected with probability 1.
	When $ w \in L $,
	\[
		p_{acc}(\mathcal{M}^{'}_{\epsilon},w) \ge \frac{1}{2c^{|w|}},
		\mbox{ and }
		p_{rej}(\mathcal{M}^{'}_{\epsilon},w) = \frac{\epsilon}{2c^{|w|}},
	\]
	and so the input is accepted with error bound $ \epsilon > 0 $ due to Lemma \ref{lemma:prej-over-pacc}, since
    \[ \frac{p_{rej}(\mathcal{M}^{'}_{\epsilon},w)}{p_{acc}(\mathcal{M}^{'}_{\epsilon},w)} \le \epsilon. \]
    Since $ p_{halt}(\mathcal{M}^{'}_{\epsilon},w) $ is always greater than $ \frac{\epsilon}{2c^{|w|}} $,
    the expected runtime of $ \mathcal{M}^{'}_{\epsilon} $ is $ O(\frac{1}{\epsilon}c^{|w|}|w|) $.
\end{proof}
\begin{lemma}
	\label{lemma:kwqfa-to-1qfarestart-constant}
	If a language $ L $ is recognized by a 1KWQFA $ \mathcal{M} $ with positive (negative)
	one-sided bounded error such that $ g_{\mathcal{M}}(n) \geq c^{-1} $ for some $ c>1 $,
	then for all $ \epsilon \in (0,\frac{1}{2}) $, $ L $ is recognized by some 
	1QFA$ ^{\circlearrowleft} $ having three more states than $ \mathcal{M} $ with positive (negative) 
	one-sided error $ \epsilon $
	in expected time $ O(\frac{1}{\epsilon}c|w|) $.	
\end{lemma}
\begin{proof}
	The construction is almost identical to that in Lemma \ref{lemma:kwqfa-to-1qfarestart-exponential}, 
	except that $ \mathsf{path_{2}} $ rejects with amplitude $ \sqrt{\epsilon} $, and
	restarts the computation with amplitude $ \sqrt{1-\epsilon} $ immediately on the left end-marker,
	 thereby causing every input to be rejected with the constant probability $ \frac{\epsilon}{2c} $.
	Hence, the expected runtime of $ \mathcal{M}^{'}_{\epsilon} $ turns out to be $ O(\frac{1}{\epsilon}c|w|) $.
\end{proof}
Lemma \ref{lemma:kwqfa-to-1qfarestart-exponential} is a useful step towards an eventual
characterization of the class of languages that are recognized with one-sided bounded error by 
1QFA$ ^{\circlearrowleft} $'s, since full classical characterizations are known \cite{YS09C} for the
classes of languages recognized by one-sided unbounded error by several 1QFA models, including the 1KWQFA.

A language $L$ is said to belong to the class S$ ^{=}_{rat} $  \cite{Tu69b,Ma93}
if there exists a rational 1PFA that accepts all and only the members of $ L $ with probability $ \frac{1}{2} $. 
\begin{theorem}
	\label{theorem:1qfa-restart-s=rat}
	For every language $L$ $ \in $ S$ ^{=}_{rat} $, there exists a number $n$ such that for all error bounds 
	$ \epsilon > 0 $, there exist $n$-state 1QFA$ ^{\circlearrowleft} $'s that recognize  
	$L$ and $\overline{L} $ with one-sided error bounded by $ \epsilon $.
\end{theorem}
\begin{proof}
	For a language  $ L $ in S$ ^{=}_{rat} $, let $ \mathcal{P} $ be the rational 1PFA associated by $ L $ 
	as described above. Turakainen \cite{Tu69b} showed that there exists a constant $ b > 1 $ such that 
	for any string $ w \notin L $,
	the probability that $ \mathcal{P} $ accepts $w$
	cannot be in the interval $ (\frac{1}{2}-b^{-|w|},\frac{1}{2}+b^{-|w|}) $.
	By using the method described in \cite{YS09C}, we can convert $ \mathcal{P} $ to a 1KWQFA $ \mathcal{M} $
	recognizing $ \overline{L} $ with one-sided unbounded error, so that $ \mathcal{M} $
accepts any $ w \in \overline{L} $ with probability greater than $ c^{-|w|} $,  for a constant $ c > b $.
	We can conclude with Lemma \ref{lemma:kwqfa-to-1qfarestart-exponential}.
\end{proof}
S$ ^{=}_{rat} $ contains many well-known languages, such as, 
$ L_{eq} $, 
$ L_{pal} = \{w \mid w = w ^{R} \} $, 
$ L_{twin} =\{ wcw \mid w \in \{a,b\}^{*} \} $,
$ L_{mult}=\{x \# y \# z \mid x,y,z \mbox{ are natural numbers in binary notation and } x \times y = z \} $, 
$ L_{square}= \{a^{n}b^{n^{2}} \mid n > 0 \} $,
$ L_{power} = \{ a^{n}b^{2^{n}} \} $, and
all \textit{polynomial languages},  \cite{Tu82} defined as
\[ \{a_{1}^{n_{1}} \cdots a_{k}^{n_{k}} b_{1}^{p_{1}(n_{1},\cdots,n_{k})} \cdots 
	b_{r}^{p_{r}(n_{1},\cdots,n_{k})} \mid p_{i}(n_{1},\cdots,n_{k}) \ge 0 \}, \]
where $ a_{1}, \cdots, a_{k},b_{1}, \cdots, b_{r} $ are distinct symbols, and 
each $ p_{i} $ is a polynomial with integer coefficients.
Note that Theorem \ref{theorem:1qfa-restart-s=rat} and 
Lemma \ref{lemma:1qfareset-simulatedby-2qcfa} answer a question posed by Ambainis and
Watrous \cite{AW02} about whether $ L_{square} $ and $ L_{power} $ can be recognized
with bounded error by 2QCFA's affirmatively.
\begin{corollary}
      The class of languages recognized by 1QFA$  ^{\circlearrowleft} $'s with bounded error
      properly contains the class of languages recognized by 1PFA$ ^{\circlearrowleft} $'s.
\end{corollary}
\begin{proof}
	This follows from Theorems \ref{theorem:pfa-restart-simulated-by-qfa-restart} and 
	\ref{theorem:1qfa-restart-s=rat}, Lemma
	\ref{lemma:restart-time}, and the fact \cite{DS92} that $ L_{pal} $ cannot be recognized with
	bounded error by 2PFA's.
\end{proof}
Since general 1QFA's are known to be equivalent in language
recognition power to 1PFA's, one has to consider a two-way model to
demonstrate the superiority of quantum computers over classical ones.
The 2QCFA is known \cite{AW02} to be superior to its
classical counterpart, the 2PFA, also by virtue of $ L_{pal} $.
Recall that, by Lemma \ref{lemma:1qfareset-simulatedby-2qcfa}, 2QCFA's
can simulate
1QFA$  ^{\circlearrowleft} $'s easily, and we do not know of a
simulation in the other direction. 

%-----------------------------------------------------------------------------%
\section{Conciseness of 2QFA's with mixed states and 2PFA's} \label{section:Conciseness}
%-----------------------------------------------------------------------------%
In this section, we demonstrate several infinite families of regular
languages which can be recognized with some fixed probability greater
than $ \frac{1}{2} $ by
just tuning the transition amplitudes of a 1QFA$ ^{\circlearrowleft} $
with a constant number of states, whereas the sizes of the
corresponding 1QFA's, 1PFA's, and 2NFA's grow without bound. One of
our constructions can be
adapted easily to
show that 1PFA$ ^{\circlearrowleft} $'s, (and, equivalently, 2PFA's),
also possess the same advantage over those machines.

\begin{definition} For an alphabet $ \Sigma $ containing symbol $ a $,
   and $ m \in \mathbb{Z}^{+} $, the family of languages $ A_{m} $
is defined as
   \[ A_{m}=\{ ua \mid u \in \Sigma^{*}, |u| \le m \}. \]
\end{definition}
Note that Ambainis et al. \cite{ANTV02} report that any Nayak
one-way quantum finite automaton\footnote{This is a 1QFA model of
intermediate power, subsuming the 1KWQFA, but strictly weaker
than the  most general models (\cite{Pa00,Ci01},
and one-way versions of 2QCFA's,) which recognize any regular
language with at most the same state cost as the corresponding DFA.}
that recognizes $ A_m $ with some fixed probability greater than $
\frac{1}{2} $ has $ 2^{\Omega(m)} $ states.

\begin{theorem}
	\label{theorem:A_m}
	$ A_{m} $  is recognized by a 6-state 1QFA$ ^{\circlearrowleft} $ $ \mathcal{M}_{m,\epsilon} $
	for any error bound $ \epsilon >0  $. Moreover, the expected runtime of $ \mathcal{M}_{m,\epsilon} $ 
	on input $ w $ is $ O( \left( \frac{1}{\epsilon} \right)^{2m}|w|) $.	
\end{theorem}
\begin{proof}
	Let $ \mathcal{M}_{m,\epsilon}=\{Q,\Sigma,\delta,q_{0},Q_{acc},Q_{rej},Q_{restart}\} $ 
	be a 1QFA$ ^{\circlearrowleft} $ with 
	$ Q_{non}=\{ q_{0}, q_{1} \} $,
	$ Q_{acc}=\{A\} $, $ Q_{rej}=\{R\} $, $ Q_{restart}=\{I_{1},I_{2}\} $.
	$ \mathcal{M}_{m,\epsilon}$ contains the transitions
	\begin{eqnarray*}
		\mathsf{U}_{\cent} \ket{q_{0}} & = & \epsilon \ket{q_{1}} +  \epsilon^{\frac{2m+5}{2}} \ket{R} +
			\sqrt{1-\epsilon^{2}-\epsilon^{2m+5}} \ket{I_{1}} \\
		\mathsf{U}_{a} \ket{q_{0}} & = & \epsilon \ket{q_{0}} + \sqrt{\frac{1}{2}-\epsilon^{2}} \ket{I_{1}} 
			+ \frac{1}{\sqrt{2}} \ket{I_{2}} \\
		\mathsf{U}_{a} \ket{q_{1}} & = & \epsilon \ket{q_{0}} + \sqrt{\frac{1}{2}-\epsilon^{2}} \ket{I_{1}} 
			- \frac{1}{\sqrt{2}} \ket{I_{2}} \\
		\mathsf{U}_{\Sigma \setminus \{a\}} \ket{q_{0}} & = & \epsilon \ket{q_{1}} 
			+ \sqrt{\frac{1}{2}-\epsilon^{2}} \ket{I_{1}} + \frac{1}{\sqrt{2}} \ket{I_{2}} \\
		\mathsf{U}_{\Sigma \setminus \{a\}} \ket{q_{1}} & = & \epsilon \ket{q_{1}}
			+ \sqrt{\frac{1}{2}-\epsilon^{2}} \ket{I_{1}}- \frac{1}{\sqrt{2}} \ket{I_{2}} \\
        \mathsf{U}_{\dollar} \ket{q_{0}} & = & \ket{A} \\
		\mathsf{U}_{\dollar} \ket{q_{1}} & = & \ket{R} \\
	\end{eqnarray*}
	and the transitions not mentioned above can be completed easily, by
	extending each $ \mathsf{U}_{\sigma}$ to be unitary.

	On the left end-marker, $\mathcal{M}_{m,\epsilon}$ rejects with probability $ \epsilon^{2m+5} $, goes on to 
	scan the input string with amplitude $ \epsilon $, and restarts immediately with the remaining probability.
	States $q_{0}$  and $q_{1}$ implement the check for the regular
	expression $ \Sigma^{*}a $, but the machine restarts with probability $ 1 - \epsilon^{2} $
	on all input symbols during this check.

	If $ w=u \sigma^{'} $ for $ u \in \Sigma^{*} $, and $ \sigma^{'} \neq a $, the input is rejected with 
	probability $1$, since $ p_{acc}(\mathcal{M}_{m,\epsilon},w)=0 $.

	If $ w=ua $ for $ u \in \Sigma^{*} $,
	\[
		p_{acc}(\mathcal{M}_{m,\epsilon},w)= \epsilon^{2|w|+2}, ~~
		p_{rej}(\mathcal{M}_{m,\epsilon},w)=\epsilon^{2m+5}.
	\]
	Hence, if $ w \in A_{m} $,
	\[ p_{acc}(\mathcal{M}_{m,\epsilon},w) \ge \epsilon^{2m+4}, \]
	and if $ w \notin A_{m} $,
	\[ p_{acc}(\mathcal{M}_{m,\epsilon},w) \le \epsilon^{2m+6}. \]
	In both cases, the corresponding ratio 
	$ \frac{p_{rej}(\mathcal{M}_{m,\epsilon},w)}{p_{acc}(\mathcal{M}_{m,\epsilon},w)} $ or 
	$ \frac{p_{acc}(\mathcal{M}_{m,\epsilon},w)}{p_{rej}(\mathcal{M}_{m,\epsilon},w)}$ is not greater than
	$ \epsilon $. 
	Thus, by Lemma \ref{lemma:prej-over-pacc}, 
	we conclude that $ \mathcal{M}_{m,\epsilon} $ recognizes $A_{m} $ with error bounded by $ \epsilon $.
	Since $ p_{halt}(\mathcal{M}_{m,\epsilon},w) $ is always greater than $ \epsilon^{2m+5} $,
	the expected runtime of $ \mathcal{M}_{m,\epsilon} $ is $ O( \left( \frac{1}{\epsilon} \right)^{2m}|w|) $.
\end{proof}

By a theorem of Rabin \cite{Ra63}, for any fixed error bound, if a
language $L$ is recognized with bounded error by
a 1PFA with $n$ states, then there exists a deterministic finite
automaton (DFA) that recognizes $L$ with
$ 2^{O(n)} $ states.
Parallelly, Freivalds et al. \cite{FOM09} note that one-way quantum
finite automata with mixed states are
no more than superexponentially more concise than DFA's.
These facts can be used to conclude that a collection of 1PFA's (or
1QFA's) with a fixed common number of
states that recognize an infinite family of languages with a fixed common
error bound less than $ \frac{1}{2} $, \textit{à la} the two-way
quantum automata of Theorem
\ref{theorem:A_m}, cannot exist, since that would imply the existence
of a similar family of DFA's of fixed size.
By the same reasoning, the existence of such families of 2NFA's can
also be overruled.

The reader should note that there exists a bounded-error 1PFA$ ^{\circlearrowleft} $
(and therefore, a 2PFA\footnote{See Section \ref{section:1pfaR-vs-2pfa}
for an examination of the relationship between the
computational powers of the 1PFA$ ^{\circlearrowleft} $
and the 2PFA.},) for $ A_{m} $, which one can obtain simply by replacing
each transition amplitude of 1QFA$
^{\circlearrowleft} $ $ \mathcal{M}_{m,\epsilon} $
defined in Theorem \ref{theorem:A_m} by the square of its modulus.
This establishes the fact that 2PFA's also possess the succinctness
advantage discussed above over 1PFA's, 1QFA's and 2NFA's.

We proceed to present two more examples.
\begin{definition}
	\label{definition:B_m}
	For $ m \in \mathbb{Z}^{+} $, the language family $ B_{m} \subseteq \{a\}^{*} $ is defined as
	\[ B_{m}=\{a^{i} \mid i \mod(m) \equiv 0 \}. \]
\end{definition}
\begin{theorem}
	\label{theorem:B_m}
	For any error bound $ \epsilon > 0 $, there exists a
	7-state 1QFA$ ^{\circlearrowleft}$ $ \mathcal{M}_{m,\epsilon} $ which
	accepts any $ w \in B_{m} $ with certainty,
	and rejects any $ w \notin B_{m} $ with probability at least $ 1-\epsilon $.
	Moreover, the expected runtime of $ \mathcal{M}_{m,\epsilon} $ on $ w $ is
	$ O \left(\frac{1}{\epsilon}\sin^{-2}(\frac{\pi}{m})|w| \right) $.
\end{theorem}
\begin{proof}
	We will construct a $ 4 $-state 1KWQFA recognizing $ \overline{B_{m}} $ with positive one-sided unbounded error,
	as described in \cite{AF98}.
	Let $ \mathcal{M}_{m} = (Q,\Sigma,\delta,q_{0},Q_{acc},Q_{rej}) $ be 1KWQFA with
	$ Q_{non}=\{q_{0},q_{1}\} $, $ Q_{acc}=\{A\} $, $ Q_{rej}=\{R\} $.
	$ \mathcal{M}_{m} $ contains the transitions
	\begin{eqnarray*}
		\mathsf{U}_{\cent} \ket{q_{0}} & = & \ket{q_{0}} \\
		\mathsf{U}_{a} \ket{q_{0}} & = & \cos(\frac{\pi}{m})\ket{q_{0}}+\sin(\frac{\pi}{m})\ket{q_{1}} \\
		\mathsf{U}_{a} \ket{q_{1}} & = & -\sin(\frac{\pi}{m})\ket{q_{0}}+\cos(\frac{\pi}{m})\ket{q_{1}} \\
		\mathsf{U}_{\dollar}\ket{q_{0}} & = & \ket{R} \\
		\mathsf{U}_{\dollar}\ket{q_{1}} & = & \ket{A},
	\end{eqnarray*}
	and the transition amplitudes not listed above are filled in to satisfy unitarity.
	$ \mathcal{M}_{m} $ begins computation at the
	$ \ket{q_{0}} $-axis, and performs a rotation by angle $ \frac{\pi}{m} $ 
	in the $ \ket{q_{0}} $-$ \ket{q_{1}} $ plane for each $ a $ it reads.
	Therefore, the value of the gap function, $ g_{\mathcal{M}_{m}} $, is not less than $ \sin^{2}(\frac{\pi}{m}) $
	for $ |w|>0 $.
	By Lemma \ref{lemma:kwqfa-to-1qfarestart-constant}, there exists a $ 7 $-state 
	1QFA$ ^{\circlearrowleft}$ $ \mathcal{M}_{m,\epsilon} $ recognizing $ \overline{B_{m}} $
	with positive one-sided bounded error and whose expected runtime is 
	$ O \left(\frac{1}{\epsilon}\sin^{-2}(\frac{\pi}{m})|w| \right) $.
	By swapping the accepting and rejecting states of $ \mathcal{M}_{m,\epsilon} $,
	we can get the desired machine.
\end{proof}

\begin{definition} 
	For an alphabet $ \Sigma $, and $ m \in \mathbb{Z}^{+} $, 
	the language family $ C_{m} $ is defined as
	 \[ C_{m}=\{ w \in \Sigma^{*} \mid |w| = m \}. \]
\end{definition}
\begin{theorem}
	\label{theorem:C_m}
	For any error bound $ \epsilon > 0 $, there exists a 7-state
	1QFA$ ^{\circlearrowleft} $ $ \mathcal{M}_{m,\epsilon} $ which accepts any $ w \in C_{m} $ with certainty, and
	rejects any $ w \notin C_{m} $ with probability at least $ 1-\epsilon $.
	Moreover, the expected runtime of $ \mathcal{M}_{m,\epsilon} $ on $ w $ is
	$ O(\frac{1}{\epsilon}2^{m}|w|) $.
\end{theorem}
\begin{proof}
	We will contruct a $ 4 $-state 1KWQFA recognizing $ \overline{C_{m}} $ with positive one-sided unbounded error.
	Let $ \mathcal{M}_{m} = (Q,\Sigma,\delta,q_{0},Q_{acc},Q_{rej}) $ be 1KWQFA with
	$ Q_{non}=\{q_{0},q_{1}\} $, $ Q_{acc}=\{A\} $, $ Q_{rej}=\{R\} $.
	$ \mathcal{M}_{m} $ contains the transitions
	\begin{eqnarray*}
		\mathsf{U}_{\cent}\ket{q_{0}} & = &  \frac{1}{\sqrt{2}} \ket{q_{0}} + 
			\left( \frac{1}{\sqrt{2}} \right)^{m+1} \ket{q_{1}} + 
			\sqrt{\frac{1}{2}-\left( \frac{1}{2} \right)^{m+1} } \ket{R} \\
		\mathsf{U}_{\sigma \in \Sigma} \ket{q_{0}} & = & \frac{1}{\sqrt{2}} \ket{q_{0}}
			+ \frac{1}{\sqrt{2}} \ket{R} \\
		\mathsf{U}_{\sigma \in \Sigma} \ket{q_{1}} & = & \ket{q_{1}} \\
		\mathsf{U}_{\dollar}\ket{q_{0}} & = & \frac{1}{\sqrt{2}} \ket{A} + \frac{1}{\sqrt{2}} \ket{R} \\
		\mathsf{U}_{\dollar}\ket{q_{1}} & = & -\frac{1}{\sqrt{2}} \ket{A} + \frac{1}{\sqrt{2}} \ket{R} \\
	\end{eqnarray*}
	with the amplitudes of the transitions not mentioned above filled in to ensure unitarity. 

	$ \mathcal{M}_{m}$ encodes the length of the input string in the amplitude of state $ q_{0} $, 
	which equals $ \left( \frac{1}{\sqrt{2}} \right)^{|w|+1} $ just before the processing of the right end-marker.
	The desired length $m$ is ``hardwired" into the amplitudes of $ q_{1} $.
	For a given input string $ w \in \Sigma^{*} $, if $ w \in C_{m} $, 
	then the amplitudes of states $ q_{0} $ and $ q_{1} $ are equal, and
	the quantum Fourier transform (QFT) \cite{KW97} performed on 
	the right end-marker sets the amplitude of $ A $ to 0. 
	Therefore, $ w $ is rejected with certainty.
	If $ w \in \overline{C_{m}} $, then the accepting probability is equal to
	\[
		\left( \left( \frac{1}{\sqrt{2}} \right)^{|w|+2} -
		 \left( \frac{1}{\sqrt{2}} \right)^{m+2} \right)^{2},
	\]
	and it is minimized when $ |w|=m+1 $, which gives us the inequality
	\[  g_{\mathcal{M}_{m}}(w) > \left( \frac{1}{2} \right)^{m+6}. \]
	By Lemma \ref{lemma:kwqfa-to-1qfarestart-constant}, there exists a $ 7 $-state 
	1QFA$ ^{\circlearrowleft}$ $ \mathcal{M}_{m,\epsilon} $ recognizing $ \overline{C_{m}} $
	with positive one-sided bounded error and whose expected runtime is 
	$ O \left(\frac{1}{\epsilon}2^{m}|w| \right) $.
	By swapping the accepting and rejecting states of $ \mathcal{M}_{m,\epsilon} $,
	we can get the desired machine.
\end{proof}

Note that, unlike what we had with Theorem \ref{theorem:A_m}, the QFA's of Theorems \ref{theorem:B_m}
and \ref{theorem:C_m} cannot be converted so easily to 2PFA's.
In fact, we can prove that there exist no 2PFA families of fixed size which recognize $ B_m $
and $ C_m $ with fixed one-sided error less than $ \frac{1}{2} $, like those QFA's:
Assume that such a 2PFA family exists. Switch the accept and reject states to obtain a family for the complements
of the languages. The 2PFA's thus obtained operate with cutpoint 0.
Obtain an equivalent 2NFA with the same number of states by converting all 
transitions with nonzero weight to nondeterministic transitions.
But there are only finitely many 2NFA's of this size, meaning that 
they cannot recognize our infinite family of languages.

%-----------------------------------------------------------------------------%
\section{Efficient Probability Amplification} \label{section:efficient-probability-amplification}
%-----------------------------------------------------------------------------%
Many automaton descriptions in this paper, and elsewhere in the theory
of probabilistic and quantum automata, describe not a single
algorithm, but a general template which one can use for building a
machine $ M_{\epsilon} $ that operates with a desired error bound $ \epsilon $.
The dependences of the runtime and number of states of $ M_{\epsilon}
$ on $ \frac{1}{\epsilon} $
are measures of the complexity of the probability
amplification process involved in the construction method used.
Viewed as such, the constructions described in the theorems in
Section \ref{section:Conciseness} are maximally efficient in terms of the state cost,
with no dependence on the error bound. In this section, we present improvements over 
previous results about the efficiency of probability amplification in 2QFA's.

%*******************************************************************%
\subsection{Improved algorithms for $ L_{eq} $} \label{subsection:2QFA-L_eq}
%*******************************************************************%

In classical computation, one only needs to sequence $ O(\log(\frac{1}{\epsilon})) $
identical copies of a given probabilistic automaton with one sided
error $ p < 1 $ to run on the same
input in order to obtain a machine with error bound $ \epsilon $.
Yakary{\i}lmaz and Say \cite{YS09B} noted that this method of
probability amplification
does not yield efficient results for 2KWQFA's; the number of machine
copies required to reduce the error to $\epsilon$
can be as high as $ ( \frac{1}{\epsilon} )^2 $.
The most succinct 2KWQFA's for $ L_{eq} $
produced by alternative methods developed in \cite{YS09B}
have $ O(\log^{2}(\frac{1}{\epsilon})\log\log(\frac{1}{\epsilon})) $ states,
and runtime linear in the size of the input $ w $.
In Appendix \ref{appendix:L_eq}, we present a construction  which yields (exponential
time) 1QFA$ ^{\circlearrowleft} $'s that recognize $ L_{eq} $ within any desired error bound
$ \epsilon $, with no dependence of the state set size on $ \epsilon $.
Ambainis and Watrous \cite{AW02} present a method which can be used to build
2QCFA's that recognize $ L_{eq} $ also with constant state set size,
where the ``tuning" of the automaton for a particular
error bound is achieved by setting some transition amplitudes
appropriately, and the expected runtime of those machines is $ O(|w|^4) $.
We now show that the 2QFA$ ^{\circlearrowleft} $ formalism allows more
efficient probability amplification.
\begin{theorem}
	\label{theorem:2qfarestart_anbn}
	There exists a constant $n$, such that, for any $ \epsilon>0 $, an $n$-state 2QFA$ ^{\circlearrowleft} $
	which recognizes $ L_{eq} $ with one-sided error bound $ \epsilon $ 
	within $ O(\frac{1}{\epsilon}|w|) $ expected runtime can be constructed, where $ w $ is the input string.
\end{theorem}
\begin{proof}
	We start with Kondacs and Watrous' original 2KWQFA \cite{KW97} $ M_N $, 
	which recognizes $ L_{eq} $ with one-sided error $ \frac{1}{N} $, for any integer $ N>1 $. 
	After a deterministic test for membership of $ a^{*}b^{*} $, 
	$ M_N $ branches to $ N $ computational paths, each of which perform a QFT at the end of the computation. 
	Set $ N=2 $. $ M_2 $ accepts all members of $ L_{eq} $ with probability 1.
	Non-members of $ L_{eq} $ are rejected with probability at least $ \frac{1}{2} $.
	We convert $ M_2 $ to a 2QFA$ ^{\circlearrowleft} $ $ \mathcal{M}_{\epsilon}^{'} $ 
	by changing the target states of the QFT as follows:
	\[
		\mathsf{path_{1}} \rightarrow \frac{1}{\sqrt{2}}\ket{\mbox{Reject}}
		+ \sqrt{\frac{\epsilon}{2}} \ket{\mbox{Accept}}+ \sqrt{\frac{1-\epsilon}{2}}\ket{\mbox{Restart}}
	\]
	\[
		\mathsf{path_{2}} \rightarrow -\frac{1}{\sqrt{2}}\ket{\mbox{Reject}}
		+ \sqrt{\frac{\epsilon}{2}} \ket{\mbox{Accept}}+ \sqrt{\frac{1-\epsilon}{2}}\ket{\mbox{Restart}}
	\]
	where the amplitude of each path is $ \frac{1}{\sqrt{2}} $.
	For a given input $ w \in \Sigma^{*} $, 
	\begin{enumerate}
		\item if $ w $ is not of the form $ a^{*}b^{*} $, then $ p_{rej}(\mathcal{M}_{\epsilon}^{'},w)=1 $;
		\item if $ w $ is of the form $ a^{*}b^{*} $ and $ w \notin L $, 
			then $ p_{rej}(\mathcal{M}_{\epsilon}^{'},w)=\frac{1}{2} $, and 
			$ p_{acc}(\mathcal{M}_{\epsilon}^{'},w)=\frac{\epsilon}{2} $;
		\item if $ w \in L $, then $ p_{rej}(\mathcal{M}_{\epsilon}^{'},w)=0 $ and 
			$ p_{acc}(\mathcal{M}_{\epsilon}^{'},w)=\epsilon $.
	\end{enumerate}
	It is easily seen that the error is one-sided. 
	Since $ \frac{p_{acc}(\mathcal{M}_{\epsilon}^{'},w)}{p_{rej}(\mathcal{M}_{\epsilon}^{'},w)} = \epsilon $,
	we can conclude with Lemma \ref{lemma:prej-over-pacc}.
	Moreover, the minimum halting probability occurs in the third case above, and so
	the expected runtime of $ \mathcal{M}_{\epsilon}^{'} $ is $ O(\frac{1}{\epsilon}|w|) $.
\end{proof}
\begin{theorem}
	\label{theorem:2qfareset_anbn}
	For any $ \epsilon \in (0,\frac{1}{2}) $, there exists a 2QFA$ ^{\curvearrowleft} $ with
	$ O(\log(\frac{1}{\epsilon})) $ states that recognizes $ L_{eq} $ with 
	one-sided error bound $ \epsilon $ in $ O(\log(\frac{1}{\epsilon})|w|) $ steps, 
	where $ w $ is the input string.
\end{theorem}
\begin{proof}
	Let $ M_{2} $ be the 2KWQFA recognizing $ L_{eq} $ with one-sided error bound $ \frac{1}{2} $ 
	mentioned in the proof of Theorem \ref{theorem:2qfarestart_anbn}.
	Then, a 2QFA$ ^{\curvearrowleft} $ that is constructed by sequentially connecting
	$ O(\log(\frac{1}{\epsilon})) $ copies of $ M_{2} $, so that the input is accepted only if it is accepted by
	all the copies, and rejected otherwise, can recognize $ L_{eq} $ with one-sided error bound $ \epsilon $.
\end{proof}

%*******************************************************************%
\subsection{An improved algorithm for $ L_{pal} $} \label{subsection:L_pal}
%*******************************************************************%

Ambainis and Watrous \cite{AW02} present a 2QCFA construction which
decides $ L_{pal} $ in expected time $ O( \left( \frac{1}{\epsilon} \right)^{|w|}|w|) $
with error bounded by $ \epsilon > 0 $, where $ w $ is the input string.
(Watrous \cite{Wa98} describes a 2KWQFA which accepts all members of the
complement of $ L_{pal} $ with probability 1, and fails to halt for all palindromes;
it is not known if 2KWQFA's can recognize this language by halting for all inputs.)
We will now present a 1QFA$ ^{\circlearrowleft} $ construction, which,
by Lemma \ref{lemma:1qfareset-simulatedby-2qcfa},
can be adapted to yield 2QCFA's with the same complexity, which reduces
the dependence of the Ambainis-Watrous method on the desired error
bound considerably.

\begin{theorem}
	\label{theorem:1qfa-restart-palindrome}
	For any $ \epsilon>0 $, there exists a 15-state 1QFA$ ^{\circlearrowleft} $ $ \mathcal{M}_{\epsilon} $
	which accepts any $ w \in L_{pal} $ with certainty, and rejects any $ w \notin L_{pal} $
	with probability at least $ 1-\epsilon $.
	Moreover, the expected runtime of $ \mathcal{M}_{\epsilon} $ on $ w $ is $ O(\frac{1}{\epsilon}3^{|w|}|w|) $.
\end{theorem}
\begin{proof}
	We will first construct a modified version of the 1KWQFA algorithm of L\={a}ce et al. \cite{LSF09} 
	for recognizing the nonpalindrome language.
	The idea behind the construction is that we encode both the input 
	string and its reverse into the amplitudes of two of the states of the machine, 
	and then perform a substraction between these amplitudes using the QFT \cite{LSF09}.
	If the input is not a palindrome, the two amplitudes do not cancel each other completely, 
	and the nonzero difference is transferred to an accept state. 
	Otherwise, the accepting probability will be zero.

	Let $ \mathcal{M} = (Q,\Sigma,\delta,q_{0},Q_{acc},Q_{rej}) $ be 1KWQFA with
	$ Q_{non}=\{p_{1},p_{2},q_{0},q_{1},q_{2},q_{3}\} $ ,
	$ Q_{acc}=\{A\} $,
	$ Q_{rej}=\{R_{i} \mid 1 \le i \le 5\} $.
	The transition function of $ \mathcal{M} $ is shown in Figure \ref{figure:pal1QFARestart}.	
	As before, we assume that the transitions not specified in the figure are filled in to ensure that
	the $ \mathsf{U}_{\sigma} $ are unitary.
	 \begin{figure}[here]
		\caption{Specification of the transition function of $ \mathcal{M} $}
		\setlength{\extrarowheight}{5pt}
		\small{
		\begin{center}
		\begin{tabular}{|c|l|l|}
			\hline
			Paths & \multicolumn{1}{c|}{$ \mathsf{U}_{\cent}, \mathsf{U}_{a} $} &
			\multicolumn{1}{c|}{$ \mathsf{U}_{b} $} \\
			\hline
			& 
			$ \mathsf{U}_{\cent} \ket{q_{0}} = \frac{1}{\sqrt{2}} \ket{p_{1}} +  \frac{1}{\sqrt{2}} \ket{q_{1}} $
			& \\
			\hline
			$ \mathsf{path_{1}} $ 
			&
			$ \begin{array}{@{}l@{}}
				\mathsf{U}_{a} \ket{p_{1}} = \sqrt{\frac{2}{3}} \ket{p_{1}} - \frac{1}{\sqrt{3}} \ket{R_{1}} \\
				\mathsf{U}_{a} \ket{p_{2}} = \frac{1}{\sqrt{6}} \ket{p_{1}} + \frac{1}{\sqrt{6}} \ket{p_{2}}
				+ \frac{1}{\sqrt{3}} \ket{R_{1}} + \frac{1}{\sqrt{3}} \ket{R_{2}}
			\end{array} $
			&
			$ \begin{array}{@{}l@{}}
				\mathsf{U}_{b} \ket{p_{1}} = \frac{1}{\sqrt{6}} \ket{p_{1}} + \frac{1}{\sqrt{6}} \ket{p_{2}}
				+ \frac{1}{\sqrt{3}} \ket{R_{1}} + \frac{1}{\sqrt{3}} \ket{R_{2}} \\
				\mathsf{U}_{b} \ket{p_{2}} = \sqrt{\frac{2}{3}} \ket{p_{2}} -  \frac{1}{\sqrt{3}} \ket{R_{1}} \\
			\end{array} $
			\\
			\hline
			$ \mathsf{path_{2}}$
             &
             $ \begin{array}{@{}l@{}}
				\mathsf{U}_{a} \ket{q_{1}} = \frac{1}{\sqrt{6}} \ket{q_{1}} + \frac{1}{\sqrt{6}} \ket{q_{3}}
				-\frac{1}{\sqrt{3}} \ket{R_{3}} + \frac{1}{\sqrt{3}} \ket{R_{4}} \\
				\mathsf{U}_{a} \ket{q_{2}} = \sqrt{\frac{2}{3}} \ket{q_{2}} + \frac{1}{\sqrt{3}} \ket{R_{5}} \\
				\mathsf{U}_{a} \ket{q_{3}} = \sqrt{\frac{2}{3}} \ket{q_{3}} + \frac{1}{\sqrt{3}} \ket{R_{3}} \\
			\end{array} $
			&
			$ \begin{array}{@{}l@{}}
				\mathsf{U}_{b} \ket{q_{1}} =\frac{1}{\sqrt{6}} \ket{q_{1}} + \frac{1}{\sqrt{6}} \ket{q_{2}}
				-\frac{1}{\sqrt{3}} \ket{R_{3}} + \frac{1}{\sqrt{3}} \ket{R_{4}} \\
				\mathsf{U}_{b} \ket{q_{2}} = \sqrt{\frac{2}{3}} \ket{q_{2}} + \frac{1}{\sqrt{3}} \ket{R_{3}} \\
				\mathsf{U}_{b} \ket{q_{3}} = \sqrt{\frac{2}{3}} \ket{q_{3}} + \frac{1}{\sqrt{3}} \ket{R_{5}} \\
			\end{array} $
			\\
			\hline
			&
			\multicolumn{2}{c|}{$ \mathsf{U}_{\dollar} $}
			\\
			\hline
			$ \mathsf{path_{1}} $
			&
			\multicolumn{2}{l|}{
				$ \begin{array}{@{}l@{}}
					\mathsf{U}_{\dollar}\ket{p_{1}} = \ket{R_{1}} \\
					\mathsf{U}_{\dollar}\ket{p_{2}} = \frac{1}{\sqrt{2}} \ket{A} + \frac{1}{\sqrt{2}} \ket{R_{2}} \\
				\end{array} $
			}
			\\
			\hline
			$ \mathsf{path_{2}} $
			&
			\multicolumn{2}{l|}{
				$ \begin{array}{@{}l@{}}
					\mathsf{U}_{\dollar}\ket{q_{1}} = \ket{R_{3}} \\
					\mathsf{U}_{\dollar}\ket{q_{2}} = -\frac{1}{\sqrt{2}} \ket{A} + \frac{1}{\sqrt{2}} \ket{R_{2}} \\
					\mathsf{U}_{\dollar}\ket{q_{3}} = \ket{R_{4}} \\
				\end{array} $
			}
			\\
             \hline
		\end{tabular}
		\end{center}
		}
		\label{figure:pal1QFARestart}
	\end{figure}

	$ \mathsf{path_{2}} $ and $ \mathsf{path_{1}} $ encode the input
	string and its reverse \cite{Ra63,Pa71} into the amplitudes of  states $ q_{2} $ and $ p_{2} $, respectively. 
	If the input is $ w=w_{1}w_{2}\cdots w_{l} $, then the values of
	these amplitudes just before the transition associated with the right end-marker in 
	the first round are as follows:
	\begin{itemize}
		\item State $ p_{2} $ has amplitude $ \frac{1}{\sqrt{2}} \left(\sqrt{\frac{2}{3}}\right)^{|w|}
		(0.w_{l}w_{l-1}\cdots w_{1})_{2} $, and 
		\item state $ q_{2} $ has amplitude $ \frac{1}{\sqrt{2}} \left(\sqrt{\frac{2}{3}}\right)^{|w|}
		(0.w_{1}w_{2} \cdots w_{l})_{2}. $
	\end{itemize}
	The factor of $ \sqrt{\frac{2}{3}} $ is due to the ``loss" of amplitude necessitated by the fact that the
	originally non-unitary encoding matrices of \cite{Ra63,Pa71} have to be ``embedded" 
	in a unitary matrix \cite{YS09D}. Note that the symbols $ a $ and $ b $ are encoded by 0 and 1, respectively.

	If $ w \in L_{pal} $, the acceptance probability is zero.
	If $ w \in \overline{L_{pal}} $, the acceptance probability is minimized by
	strings which are almost palindromes, except for a single defect in the middle, that is,
	when $ |w|=2k $ for $ k \in \mathbb{Z}^{+} $,  $ w_{i}=w_{2k-i+1} $, where $ 1 \le i \le k-1 $, and
	$ w_{k} \neq w_{k+1} $, so,
	\[ g_{\mathcal{M}}(w) \ge \frac{1}{8} \left( \frac{1}{3} \right)^{|w|}. \]
	By Lemma \ref{lemma:kwqfa-to-1qfarestart-exponential}, there exists a $ 15 $-state 
	1QFA$ ^{\circlearrowleft}$ $ \mathcal{M}_{\epsilon} $ recognizing $ \overline{L_{pal}} $
	with positive one-sided bounded error, whose expected runtime is 
	$ O(\frac{1}{\epsilon}3^{|w|}|w|) $.
	By swapping accepting and rejecting states of $ \mathcal{M}_{m} $,
	we can get the desired machine.
\end{proof}

Note that the technique used in the proof above can be extended easily
to handle bigger input alphabets by using the matrices defined on Page 169 of \cite{Pa71},
and the method of simulating stochastic matrices by unitary matrices described in \cite{YS09D}.

%-----------------------------------------------------------------------------%
\section{1PFA$ ^{\circlearrowleft} $ vs. 2PFA} \label{section:1pfaR-vs-2pfa}
%-----------------------------------------------------------------------------%

It is interesting to examine the power of the restart move in
classical computation as well. Any 1PFA$ ^{\circlearrowleft} $ which
runs in expected $ t $ steps
can be simulated by a 2PFA which runs in expected $ 2t $ steps (see
Lemma \ref{lemma:restart-time}).
We ask in this section whether the restart move can substitute the ``left" and
``stationary" moves of a 2PFA without loss of computational power.
Since every polynomial-time 2PFA recognizes a regular language, which
can of course be recognized by using only ``right" moves, we focus on
the best-known example of a nonregular language that can be recognized
by an exponential-time 2PFA.
\begin{theorem}
       There exists a natural number $k$, such that for any $ \epsilon>0 $,
there exists a $k$-state 1PFA$ ^{\circlearrowleft} $ $
\mathcal{P}_{\epsilon} $
       recognizing language $ L_{eq} $ with error bound $ \epsilon $ and
expected runtime
       $ O( (\frac{2}{\epsilon^{2}})^{|w|}|w|) $, where $ w $ is the input string.
\end{theorem}
\begin{proof}
       We will construct the 1PFA$ ^{\circlearrowleft} $ $
\mathcal{P}_{\epsilon} $ as follows:
       Let $ x = \frac{\epsilon^{2}}{2} $.
       The computation splits into three paths called $ \mathsf{path_{1}} $,
$ \mathsf{path_{2}} $, and
       $ \mathsf{path_{3}} $ with equal probabilities on symbol $ \cent $.
       All three paths, while performing their main tasks, parallelly check
whether the input is of the form
       $ a^{*}b^{*} $, if not, all paths simply reject.
       The main tasks of the  paths are as follows:
       \begin{list}{$ \bullet $}{}
               \item $ \mathsf{path_{1}} $ moves on with probability $ x $ and
restarts with probability $ 1-x $
               when reading symbols $ a $ and $ b $. After reading the right
end-marker $ \dollar $,
               it accepts with probability with $ 1 $.
               \item $ \mathsf{path_{2}} $ moves on with probability $ x^{2} $ and
restarts with probability $ 1-x^{2} $
               when reading symbol $ a $. On $ b $'s, it continues with the ``syntax" check.
               After reading the $ \dollar $, it rejects with probability $
\frac{\epsilon}{2} $ and
               restarts with probability $ 1-\frac{\epsilon}{2} $.
               \item $ \mathsf{path_{3}} $ is similar to $ \mathsf{path_{2}} $,
except that the transitions
               of symbols $ a $ and $ b $ are interchanged.
       \end{list}

       If the input is of the form $ a^{m}b^{n} $, then the accept and
reject probabilities
       of the first round are calculated as
   \[
       p_{acc}(\mathcal{P}_{\epsilon},w)=\frac{1}{3} x^{m+n},~ \mbox{ and }~
           p_{rej}(\mathcal{P}_{\epsilon},w)= \frac{\epsilon}{6} \left(
x^{2m} + x^{2n} \right).
    \]

       If $ m = n $, then
       \[ \frac{p_{rej}(\mathcal{P}_{\epsilon},w)}{p_{acc}(\mathcal{P}_{\epsilon},w)}
= \epsilon. \]

       If $ m \neq n $ (assume without loss of generality that $ m = n+d $
for some $ d \in \mathbb{Z}^{+} $) ,
       then
    \[
            \frac{p_{acc}(\mathcal{P}_{\epsilon},w)}{p_{rej}(\mathcal{P}_{\epsilon},w)}
=
            \frac{2}{\epsilon} \frac{x^{2n+d}}{x^{2n+2d}+x^{2n}} =
            \frac{2}{\epsilon} \frac{x^{d}}{x^{2d}+1} <
            \frac{2}{\epsilon} x^{d} \le \frac{2}{\epsilon} x
       \]
       By replacing $ x=\dfrac{\epsilon^{2}}{2} $, we can get
       \[ \frac{p_{acc}(\mathcal{P}_{\epsilon},w)}{p_{rej}(\mathcal{P}_{\epsilon},w)}
< \epsilon. \]

       By using Lemma \ref{lemma:prej-over-pacc}, we can conclude that $
\mathcal{P}_{\epsilon} $
       recognizes $ L_{eq} $ with error bound $ \epsilon $.

       Since $ p_{halt}(\mathcal{P}_{\epsilon},w) $ is always greater than $
\frac{1}{3} x^{|w|}  $,
       the expected runtime of the algorithm is $ O(
(\frac{2}{\epsilon^{2}})^{|w|}|w|) $,
       where $ w $ is the input string.
\end{proof}

%-----------------------------------------------------------------------------%
\section{Concluding remarks} \label{section:ConcludingRemarks}
%-----------------------------------------------------------------------------%

By a theorem of Dwork and Stockmeyer \cite{DS90},  for a fixed $
\epsilon< \frac{1}{2} $,
if $L$ is recognized by a $O(n)$--time 2PFA with $c$ states within error
probability $ \epsilon $, then
$ L $ is also recognized by a DFA with $c^{bc^2}$ states, where the
number $b$ depends on the constant hidden in the big-$O$.
The two-way machines of Section \ref{section:Conciseness}
can be seen to have such factors that grow with $m$ in the expressions
for their time
complexities; this is how the machines described in that section achieve
their huge superiority in terms  of the state cost over the other
models that they are compared with.

It is known \cite{YS09D} that 2KWQFA's can recognize some
nonstochastic languages
(i.e. those which cannot be recognized by 2PFA's) in the unbounded
error setting.
On the other hand, we conjecture that 2QFA's with classical head
position, such as the 2QCFA, cannot recognize any nonstochastic language.
Therefore, it is an interesting question whether 2QFA$
^{\curvearrowleft} $'s (or possibly an even more general
2QFA model allowing head superposition) can recognize any
nonstochastic language with bounded error or not.

Some other open questions related to this work are:
\begin{enumerate}
  \item Can 1QFA$ ^{\circlearrowleft} $'s simulate 2QCFA's?
   \item Are 1PFA$^{\circlearrowleft}$'s (with just ``restart" and
``right" moves) equivalent
   in power to 2PFA's in the bounded-error setting, as hinted by
Section \ref{section:1pfaR-vs-2pfa}?
   \item Does there exist an analogue of the Dwork-Stockmeyer theorem
mentioned above
   for two-way quantum finite automata?
\end{enumerate}

%-----------------------------------------------------------------------------%
\section*{Acknowledgements} \label{section:Acknowledgement}
%-----------------------------------------------------------------------------%
We are grateful to Andris Ambainis and John Watrous for their helpful
comments on earlier versions of this paper. We also thank
R\={u}si\c{n}\v{s} Freivalds for kindly providing us a copy of
reference \cite{LSF09}.

\bibliographystyle{alpha}
\bibliography{YakaryilmazSay}

\appendix 

%*******************************************************************%
\section{A 1QFA$ ^{\circlearrowleft} $ algorithm for $ L_{eq} $} \label{appendix:L_eq}
%*******************************************************************%

\begin{theorem}
	For any $ \epsilon>0 $, there exists a 15-state 1QFA$ ^{\circlearrowleft} $ $ \mathcal{M}_{\epsilon} $,
	which accepts any $ w \in L_{eq} $ with certainty, and rejects any $ w \notin L_{eq} $
	with probability at least $ 1-\epsilon $.
	Moreover, the expected runtime of $ \mathcal{M}_{\epsilon} $ on $ w $ is 
	$ O(\frac{1}{\epsilon}(2\sqrt{2})^{|w|}|w|) $ .
\end{theorem}
\begin{proof}
	We will contruct a $ 12 $-state 1KWQFA recognizing $ \overline{L_{eq}} $ with positive one-sided unbounded error.
	Let $ \mathcal{M} = (Q,\Sigma,\delta,q_{0},Q_{acc},Q_{rej}) $ be 1KWQFA with
	$ Q_{non}=\{p_{0},p_{1},p_{2},q_{0},q_{1},q_{2}\} $ ,
	$ Q_{acc}=\{A_{1},A_{2},A_{3}\}, Q_{rej}=\{R_{1},R_{2},R_{3}\} $.
	The transition function of $ \mathcal{M} $ is shown in Figure \ref{figure:anbn1QFARestart}.	
	As before, we assume that the transitions not specified in the figure are filled in to ensure that
	the $ \mathsf{U}_{\sigma} $ are unitary.

	\begin{figure}[here]
		\caption{Specification of the transition function of $ \mathcal{M} $}
		\setlength{\extrarowheight}{2pt}
		\small{
		\begin{center}
		\begin{tabular}{|c|l|l|l|}
			\hline
			Paths & \multicolumn{1}{c|}{$ \mathsf{U}_{\cent}, \mathsf{U}_{a} $} &
			\multicolumn{1}{c|}{$ \mathsf{U}_{b} $} & \multicolumn{1}{c|}{$ \mathsf{U}_{\dollar} $} \\
			\hline
			& $ \mathsf{U}_{\cent} \ket{q_{0}} = \frac{1}{\sqrt{2}} \ket{p_{0}} + \frac{1}{\sqrt{2}}\ket{q_{0}} $ 
			& & \\
			\hline
			$ \mathsf{path_{1}} $
			&
			$ \begin{array}{@{}l@{}}
				\mathsf{U}_{a} \ket{p_{0}} = \frac{1}{2} \ket{p_1} + \frac{1}{2}\ket{R_{1}} 
				+ \frac{1}{\sqrt{2}}\ket{R_{2}} \\
				\mathsf{U}_{a} \ket{p_{1}} = \frac{1}{2} \ket{p_1} + \frac{1}{2}\ket{R_{1}}
				- \frac{1}{\sqrt{2}}\ket{R_{2}} \\
				\mathsf{U}_{a} \ket{p_{2}} = \ket{A_{1}}
			\end{array} $
			&
			$ \begin{array}{@{}l@{}}
				\mathsf{U}_{b} \ket{p_{0}} = \ket{A_{1}} \\
				\mathsf{U}_{b} \ket{p_{1}} = \frac{1}{\sqrt{2}} \ket{p_{2}}+\frac{1}{\sqrt{2}}\ket{R_{1}} \\
				\mathsf{U}_{b} \ket{p_{2}} = \frac{1}{\sqrt{2}} \ket{p_{2}}-\frac{1}{\sqrt{2}}\ket{R_{1}}
			\end{array} $
			&
			$ \begin{array}{@{}l@{}}
				\mathsf{U}_{\dollar}\ket{p_{0}}=\ket{R_{1}} \\
				\mathsf{U}_{\dollar}\ket{p_{1}} = \ket{A_{1}} \\
				\mathsf{U}_{\dollar}\ket{p_{2}} = \frac{1}{\sqrt{2}} \ket{R_{2}} + \frac{1}{\sqrt{2}} \ket{A_{2}} \\
			\end{array} $
			\\
			\hline
			$ \mathsf{path_{2}}$
			&
			$ \begin{array}{@{}l@{}}
				\mathsf{U}_{a} \ket{q_{0}} = \frac{1}{\sqrt{2}} \ket{q_1} + \frac{1}{\sqrt{2}} \ket{R_{3}} \\
				\mathsf{U}_{a} \ket{q_{1}} = \frac{1}{\sqrt{2}} \ket{q_1}-\frac{1}{\sqrt{2}} \ket{R_{3}} \\
				\mathsf{U}_{a} \ket{q_{2}} = \ket{A_{2}}
			\end{array} $
			&
			$ \begin{array}{@{}l@{}}
				\mathsf{U}_{b} \ket{q_{0}} = \ket{A_{2}} \\
				\mathsf{U}_{b} \ket{q_{1}} = \frac{1}{2} \ket{q_2} + \frac{1}{2}\ket{R_{2}} 
				+ \frac{1}{\sqrt{2}}\ket{R_{3}} \\
				\mathsf{U}_{b} \ket{q_{2}} = \frac{1}{2} \ket{q_2} +
				\frac{1}{2}\ket{R_{2}} - \frac{1}{\sqrt{2}}\ket{R_{3}}
			\end{array} $
			&
			$ \begin{array}{@{}l@{}}
				\mathsf{U}_{\dollar}\ket{q_{0}} = \ket{R_{3}} \\
				\mathsf{U}_{\dollar}\ket{q_{1}} = \ket{A_{3}} \\
				\mathsf{U}_{\dollar}\ket{q_{2}} = \frac{1}{\sqrt{2}} \ket{R_{2}} - \frac{1}{\sqrt{2}} \ket{A_{2}} \\
			\end{array} $
			\\
			\hline       
		\end{tabular}
		\end{center}
		}\label{figure:anbn1QFARestart}
	\end{figure}

	As seen in the figure, $ \mathcal{M} $ branches to two paths on the left end-marker. 
	All paths rejects immediately if the input $ w \in \{a,b\}^{*} $ is the empty string, 
	and accepts with nonzero probability, say $ \alpha $, if it is of the form 
	$ (\{a,b\}^{*} \setminus a^{*}b^{*}) \cup a^{+} \cup b^{+}  $.
	Otherwise, $ w = a^{m}b^{n}$ $ (m,n>0) $, and the amplitudes of the
	paths just before the transition associated with the right end-marker in the first round are as follows:
	\begin{itemize}
		\item State $ p_{2} $ has amplitude $ \frac{1}{\sqrt{2}} (\frac{1}{2})^{m} (\frac{1}{\sqrt{2}})^{n}$,
		\item state $ q_{2} $ has amplitude $ \frac{1}{\sqrt{2}} (\frac{1}{\sqrt{2}})^{m} (\frac{1}{2})^{n}$.
	\end{itemize}
	If $ m=n $,  then the accepting probability is zero. 
	If $ m \neq n $ (assume that $ m=n+d $ for some $ d \in \mathbb{Z}^{+} $), 
	then the accepting probability is equal to
	\[ \left( \frac{1}{2} \right)^{m+n+1} \left( \left( \frac{1}{\sqrt{2}} \right)^{m} - 
	\left( \frac{1}{\sqrt{2}} \right)^{n} \right)^{2}
	=
	\underbrace{\left( \frac{1}{2} \right)^{m+2n+1}}_{ > \left( \frac{1}{2} \right)^{\frac{3|w|}{2}+1} } 
	\underbrace{\left( 1 - \left( \frac{1}{\sqrt{2}} \right)^{d-2}
	+ \left( \frac{1}{2} \right)^{d} \right)}_{ > \frac{1}{16}} 
	\]
	Since $ \alpha $ is always greater than this value,
	\[ g_{\mathcal{M}}(|w|) > \left( \frac{1}{2} \right)^{\frac{3|w|}{2}+5}, \]
	for $ |w|>0 $.
	By Lemma \ref{lemma:kwqfa-to-1qfarestart-exponential}, there exists a $ 15 $-state 
	1QFA$ ^{\circlearrowleft}$ $ \mathcal{M}_{\epsilon} $ recognizing $ \overline{L_{eq}} $
	with positive one-sided bounded error and whose expected runtime is 
	$ O(\frac{1}{\epsilon}(2\sqrt{2})^{|w|}|w|) $.
	By swapping accepting and rejecting states of $ \mathcal{M}_{m} $,
	we can get the desired machine.
\end{proof}

\end{document}